\documentclass[aip,jmp,reprint,onecolumn]{revtex4-1}

\usepackage{ytableau, graphicx}
\usepackage{amsfonts, amsmath, amsthm}
\usepackage{mathtools}

\newtheorem{theorem}{Theorem}[section]
\newtheorem{lemma}[theorem]{Lemma}

\theoremstyle{definition}
\newtheorem{definition}[theorem]{Definition}

\theoremstyle{remark}

\newcommand{\be}{\begin{equation}}
\newcommand{\ee}{\end{equation}}
\newcommand*{\vbar}[1]{\multicolumn{1}{c|}{#1}}
\newcommand*{\re}[1]{\operatorname{Re}{#1}}
\newcommand*{\im}[1]{\operatorname{Im}{#1}}
\newcommand\qbinom[3]{\genfrac{[}{]}{0pt}{}{#1}{#2}_{#3}}
\newcommand*{\sign}[1]{\operatorname{sign}{(#1)}}
\providecommand{\abs}[1]{\lvert#1\rvert}

\newcommand{\upw}[2]{{#1}^{\underline{{#2}}}}
\newcommand{\lr}[1]{\multicolumn{1}{|c}{#1}}
\newcommand{\rr}[1]{\multicolumn{1}{c|}{#1}}
\newcommand{\widest}{\ensuremath{\upw{t_2}{2}}}
\newcommand{\newWidth}[1]{\makebox[\widthof{\widest}]{\ensuremath{#1}}}

\begin{document}


\title{Evaluation and spanning sets of confluent Vandermonde forms} 



\author{D. K. Sunko}
\email{dks@phy.hr}
\thanks{This work was supported by the Croatian Science Foundation under Project No. IP-2018-01-7828.}
\affiliation{Department of Physics, Faculty of Science, University of Zagreb, Bijeni\v{c}ka cesta 32, HR-10000 Zagreb, Croatia.}


\date{\today}

\begin{abstract}
An arbitrary derivative of a Vandermonde form in $N$ variables is given as $[n_1\cdots n_N]$, where the $i$-th variable is differentiated $N-n_i-1$ times, $1\le n_i\le N-1$. A simple decoding table is introduced to evaluate it by inspection. The special cases where $0\le n_{i+1} - n_i \le 1$ for $0<i<N$ are in one-to-one correspondence with ribbon Young diagrams. The respective $N!$ standard ribbon tableaux map to a complete graded basis in the space of $S_N$-harmonic polynomials. The mapping is realized as an efficient algorithm generating any one of $N!$ bases with $N!$ basis elements, both indexed by permutations. The result is placed in the context of a geometric interpretation of the Hilbert space of many-fermion wave functions.
\end{abstract}

\maketitle

\section{Introduction}

Derivatives of the Vandermonde form have historically appeared in the context of polynomial interpolation, when one needs to fit both the value and successive derivatives of a function at a given set of points. With certain restrictions stemming from the nature of this particular problem, e.g.\ some of the points are taken to be the same, they are called confluent Vandermonde forms.\cite{Bjork70} In this work, the term confluent Vandermonde form is used for an arbitrary partial derivative of the Vandermonde form, keeping all variables distinct. Spaces spanned by these derivatives are natural objects of study in representation theory, algebraic combinatorics, and the theory of symmetric functions.\cite{Shephard54,Humphreys90,Bergeron13}

In physics, it is known that the Hilbert space of $N$-fermion wave functions in $d$ dimensions is a free module, generated by $D=N!^{d-1}$ linearly independent functions $\Psi_i$ in $d$ sets of $N$ variables:\cite{Sunko16-1}
\be
\Psi=\sum_{i=1}^D\Phi_i\Psi_i,
\label{scheme}
\ee
where the coefficients $\Phi_i$ are linear combinations of functions which factorize in each set of variables separately:
\be
\Phi_{i,1}(t_1,\ldots ,t_N)\Phi_{i,2}(u_1,\ldots ,u_N)\cdots \Phi_{i,d}(v_1,\ldots ,v_N),
\label{choice}
\ee
where the $\Phi_{i,j}$ are each symmetric in its own set of variables. This choice of coefficient ring is physically motivated by the nature of density waves in many-fermion systems.  Like the electromagnetic field, density waves can always be resolved into independent plane waves in the three space directions. Each plane wave is a boson (symmetric function) by itself, so Eq.~\eqref{choice} encodes the structure of a general physical bosonic excitation, which is less general than an arbitrary symmetric function.

Unlike the more familiar vector-space approach to many-body Hilbert space, the free-module approach is natively non-linear. Generators of the free module are geometric objects in wave-function space, which restrict fermion motion kinematically. They provide qualitative insights into the many-body wave function, e.g.\ bands in the spectra of finite systems are interpreted as ideals generated by their band-heads.\cite{Rozman20} In this way, algebraic geometry becomes the natural framework for the particle picture of quantum mechanics.\cite{Sunko20}

For $d=3$, there is a conjecture,\cite{Sunko20} obviously extending to any odd $d$, that the generators (secondary invariants) $\Psi_i$ of the free module are precisely all distinct iterated symmetrized derivatives $\prod_{a,b,c}\nabla^{(a,b,c)}\mathcal{D}_N$ of the triple product
\be
\mathcal{D}_N = 
\widetilde{\Delta}_N(t)\widetilde{\Delta}_N(u)\widetilde{\Delta}_N(v),
\label{triple}
\ee
where the symmetrized derivative is
\be
\nabla^{(a,b,c)}=\sum_{i=1}^N\frac{\partial^a}{\partial t_i^a}
\frac{\partial^b}{\partial u_i^b}\frac{\partial^c}{\partial v_i^c}
\ee
and $\Delta_N(t)$ is the Vandermonde form in $N$ variables $t_1,\ldots,t_N$, similarly for $u$ and $v$. The tilde denotes a special normalization, suited for the study of derivatives. In this normalization, the confluent Vandermonde form (cv-form) reads
\be
[n_1\cdots n_N]=\begin{vmatrix}
t_1^{n_1}/n_1!&\cdots&t_N^{n_N}/n_N!\\
t_1^{n_1-1}/(n_1-1)!&\cdots&t_N^{n_N-1}/(n_N-1)!\\
\vdots&&\vdots\\
0 \mbox{ or } 1 &\cdots& 0 \mbox{ or } 1
\end{vmatrix},\quad 0\le n_i< N.
\label{cvf}
\ee
Different permutations of the entries may give rise to linearly independent polynomials. Each entry $n_i$ means that the Vandermonde form is differentiated $N-n_i-1$ times with respect to $t_i$. In this notation, the Vandermonde form is normalized:
\be
\widetilde{\Delta}_N(t)=[N-1\cdots N-1]_{N\,\mathrm{times}} = \prod_{1\le i<j\le N} \frac{t_i-t_j}{j-i},
\label{vdm}
\ee
so that
\be
\partial_{t_1}^{N-n_1-1} \cdots \partial_{t_N}^{N-n_N-1} 
[N-1 N-1 \ldots N-1]_{N\,\mathrm{times}} = [n_1 n_2 \cdots n_N].
\ee
The usual definition of the Vandermonde form, without the denominator in Eq.~\eqref{vdm}, is denoted $\Delta_N(t)$. The (total) degree of a cv-form is the sum of its entries minus $N(N-1)/2$, noting that $[0\,1\,2\,\cdots\,N-1]=1$.

In this work, an important step to constructing the above-mentioned generating set is made. An efficient algorithm is presented which provides an explicit basis for the derivatives of a single Vandermonde form in the above notation.

The connection of this result with quantum mechanics rests on a well-known bijection\cite{Bargmann61} between the normalized Hermite functions $\psi_n(x)\in \mathrm{L}^2(\mathbb{R})$ and powers $t^n\in \mathrm{F}(\mathbb{C})$, the \emph{Bargmann space} of entire functions $F:\mathbb{C}\to\mathbb{C}$ such that
\be
\int_{\mathbb{C}}|F(t)|^2d\lambda(t)<\infty,
\ee
with scalar product
\be
(f,g)=\int_{\mathbb{C}}\bar{f}g\,d\lambda(t),\quad
d\lambda(t)=\frac{1}{\pi}e^{-|t|^2}d\re{t}\,d\im{t}.
\ee
The bijection (Bargmann transform) is just $\psi_n(x) \leftrightarrow t^n/\sqrt{n!}$. Thus expressions in the formal variables $t_i$ etc.\ can be interpreted directly in terms of many-body wave functions in the harmonic-oscillator basis. These functional-analytic connotations will not be explored here. Generally, the transform is a secure foundation to study the geometry of many-body wave-function space by classical methods. It must be extended to three sets of variables to reflect the dimension of laboratory space. In this context,\cite{Sunko16-1} one takes the indeterminates in Eq.~\eqref{triple} as $N$ points in $\mathbb{C}^3$.

This article consists of three parts, motivation, algorithm, and proofs. In the first part, the cv-form is evaluated by inspection. Interpreting the leading term in this expansion as a skew (ribbon) Young diagram provides the motivation for the algorithm. The algorithm is described in the second part, taking for granted that the standard Young skew tableaux, generated by these ribbons, encode a complete basis. The completeness is proven in the third part. The physical context is briefly revisited in the discussion at the end, where a connection to recent investigations\cite{Bergeron13} extending the so-called $N!$ and $(N+1)^{N-1}$ hypotheses\cite{Haiman03} to $d=3$ is also pointed out.

\section{Motivation}

\subsection{Laplace block expansion}

Our main tool will be Laplace's expansion into blocks,\cite{Aitken39} of which the familiar (cofactor) expansion is a special case. Consider for example a $4\times 4$ matrix $A=\{a_{ij}\}$. The standard cofactor expansion by first column can be restyled graphically:
\be
\begin{matrix}
\lvert a_{11}\rvert\cdot & \\
& \begin{vmatrix} a_{22} & a_{23} & a_{24} \\
  a_{32} & a_{33} & a_{34} \\
  a_{42} & a_{43} & a_{44} \end{vmatrix}
\end{matrix}
\begin{matrix}
-\lvert a_{21}\rvert\cdot & \\
& \begin{vmatrix} a_{12} & a_{13} & a_{14} \\
  a_{32} & a_{33} & a_{34} \\
  a_{42} & a_{43} & a_{44} \end{vmatrix}
\end{matrix}
\begin{matrix}
+\lvert a_{31}\rvert\cdot & \\
& \begin{vmatrix} a_{12} & a_{13} & a_{14} \\
  a_{22} & a_{23} & a_{24} \\
  a_{42} & a_{43} & a_{44} \end{vmatrix}
\end{matrix}
\begin{matrix}
-\lvert a_{41}\rvert\cdot & \\
& \begin{vmatrix} a_{12} & a_{13} & a_{14} \\
  a_{22} & a_{23} & a_{24} \\
  a_{32} & a_{33} & a_{34} \end{vmatrix}.
\end{matrix}
\ee
The $4\times 4$ determinant $|A|$ is understood as a sum of products of $3\times 3$ and $1\times 1$ determinants (not absolute values), which are themselves diagonal blocks of the original determinant and of those with rows \emph{shuffled} in all allowed ways, namely: by all permutations of rows which preserve the original row-order within each block. Offsetting the terms in the products is meant to suggest these diagonal blocks graphically. The sign of each contribution is the sign of the corresponding row-permutation (shuffle). The same expression in a more compact notation is
\be
\abs{A}=(1|2\,3\,4)-(2|1\,3\,4)+(3|1\,2\,4)-(4|1\,2\,3),
\ee
where the numbers refer to the row-shuffle which determines the diagonal blocks. (Interpreting them as row numbers produces an expansion in the first column, while the interpretation as column-numbers would produce an expansion in the first row.)

As noticed already by Laplace,\cite{Laplace72} any blocking will work, e.g.
\be
\abs{A}=(1\,2|3\,4)-(1\,3|2\,4)+(1\,4|2\,3)+(2\,3|1\,4)-(2\,4|1\,3)+(3\,4|1\,2),
\label{A44}
\ee
where the \emph{row-blocks} are now
\be
(i\,j|k\,l)=\begin{vmatrix}
a_{i1} & a_{i2} \\
a_{j1} & a_{j2}
\end{vmatrix}\cdot\begin{vmatrix}
a_{k3} & a_{k4} \\
a_{l3} & a_{l4}
\end{vmatrix},
\label{detshuf}
\ee
and the sign of the contribution is the sign of the row-permutation $ijkl$. One can visualize the diagonal $2\times 2$ blocks as windows, such that each shuffle of the rows brings another set of determinants into view, which are then multiplied and added with the permutation sign. The number of blocks and their ordering is arbitrary, e.g. $4=1+1+2$ gives
\be
\abs{A}=(1|2|3\,4)-(2|1|3\,4)-(1|3|2\,4)+(3|1|2\,4)+\ldots,
\ee
which can be obtained by expanding the first determinant in each row-block of the expansion~\eqref{A44}. If all blocks are $1\times 1$, every permutation is trivially an allowed shuffle, so the Laplace expansion reverts to the (defining) Leibniz formula. This observation is an idea to prove the general expression that we only state here.

Let $N=m_1+\ldots +m_k$ be a \emph{composition} of natural numbers (ordering matters). Let $U=(\mathcal{U}_1,\ldots,\mathcal{U}_k)$ be a list of $k$ disjoint sets of size $\abs{\mathcal{U}_i}=m_i$, such that $\cup_i \mathcal{U}_i=\{1,\ldots,N\}$. Let $\pi_i$ be the ascending list of numbers contained in each $\mathcal{U}_i$. A simple concatenation of these lists defines a \emph{shuffle permutation} $\pi(U):=\pi_1\ldots\pi_k$.\footnote{The standard combinatorial definition of a shuffle is the inverse of ours: a permutation is our shuffle if there exists a standard $(m_1,\ldots,m_k)$-shuffle (\emph{interleaving}) which sorts the elements.} Any fixed composition of $N$ gives a Laplace expansion into blocks for the determinant of any $N\times N$ matrix $A$:
\be
\abs{A}=\sum_U\sign{\pi(U)}A_1A_2\cdots A_k,
\label{laplace}
\ee
where the sum is over all distinct lists $U$ defined above. The $A_j$ are minors whose product is best denoted by the shuffle itself, understood as a row-shuffle here:
\be
A_1A_2\cdots A_k=(\pi_1|\pi_2|\cdots |\pi_k),
\ee
for which Eq.~\eqref{detshuf} is a transparent template. Formally, for $j=1,\ldots,k$, let $m=m_j$, $l=m_0+\cdots+m_{j-1}$ with $m_0=0$, and $\pi_j=(i_1,\ldots,i_m)$,  Then the row-shuffled minors are
\be
A_j=\begin{vmatrix}
a_{i_1,l+1} & \cdots & a_{i_1,l+m} \\
\vdots & & \vdots \\
a_{i_{m},l+1} & \cdots & a_{i_{m},l+m} 
\end{vmatrix}.
\ee
(For a column-shuffle, transpose the indices.)

The block-expansion formula is particularly efficient when the determinant contains a number of zeros, which can be permuted into the lower left corner. All shuffles which introduce a row of zeros into a block can be dropped, which may lead to a considerable reduction in the number of terms, if the block sizes are chosen to maximize this effect [e.g., use~\eqref{A44} if $a_{41}=a_{42}=0$]. As we shall see now, cv-forms are a refined case of this kind.

\subsection{Evaluation of the confluent Vandermonde form}

We evaluate the cv-form by direct attack, beginning with two notational observations. First, at least one entry $n_i$ in Eq.~\eqref{cvf} must be equal to $N-1$, otherwise the bottom row is zero. Second, zero entries can be removed, so that $1\le n_i< N$ for all $i=1,\ldots,N$.

\begin{lemma}
A zero entry in the cv-form~\eqref{cvf} can be removed by the rule
\be
[n_1\cdots n_{k-1}\,0\,n_{k+1}\cdots n_N] =
(-1)^{N-1}[n_1'\cdots n_{k-1}'\,N-1\,n_{k+1}'\cdots n_N'],
\label{zerorule}
\ee
where $n_i'=n_i-1$. Either (a) the rule will remove a zero, by iteration if necessary, or (b) the cv-form evaluates to $\pm 1$ or zero.

\emph{Example.} $[0\,1\,3\,3] = -[3\,0\,2\,2] = [2\,3\,1\,1]=t_3-t_4$.
\label{nozero}
\end{lemma}

\begin{proof}
The expansion of the left-hand side in the $k$-th column is in a single minor. The expansion of the right-hand side in the bottom row is in the same minor. The sign appears because the `$1$' has moved from first to last row. (a) If the right-hand side has no zero entries, there is nothing left to prove. (b) If two $n_i$'s are zero, the cv-form is zero (two equal columns). If the $n_i$'s are all distinct, the rule iterates in an infinite loop, but if they are increasing, the cv-form is triangular, so it evaluates to $1$. Permuting these entries gives a sign, so it evaluates to $\pm 1$. Any other outcome with a zero has exactly one entry equal to $N-1$ and at least two entries equal to some $n$, $0<n<N-1$. Choosing the smallest such $n$, iteration must terminate within the next $n$ steps, because the original zero, promoted to $N-1$, cannot become zero again before either (a) there are no more zeros, or (b) the entries equal to $n$ have become zeros, so the cv-form is zero.
\end{proof}

Assume that the cv-form entries are nondecreasing, $n_i\le n_j$ when $i<j$. When the row-blocks group equal entries $n_i$, the corresponding minors $A_i$ in Eq.~\eqref{laplace} are precisely numerators of Schur functions (Theorem 7.15.1 of Ref.~\onlinecite{Stanley99}). One can use this observation to evaluate any row-block by inspection. The algorithm is most easily described on an example.

Consider the cv-form $[2\,2\,4\,4\,5\,5]$. Explicitly, 
\be
[2\,2\,4\,4\,5\,5]=\begin{vmatrix}
\rr{\begin{matrix}
\upw{t_1}{2} & \upw{t_2}{2} \\
\upw{t_1}{1} & \upw{t_2}{1}
\end{matrix}}
&
\begin{matrix}
\upw{t_3}{4} & \upw{t_4}{4} \\
\upw{t_3}{3} & \upw{t_4}{3}
\end{matrix}
&
\begin{matrix}
\upw{t_5}{5} & \upw{t_6}{5} \\
\upw{t_5}{4} & \upw{t_6}{4}
\end{matrix}
\\\cline{1-2}
\begin{matrix}
\upw{t_1}{0} & \upw{t_2}{0} \\
0 & 0 
\end{matrix}
&
\lr{\begin{matrix}
\upw{t_3}{2} & \upw{t_4}{2} \\
\upw{t_3}{1} & \upw{t_4}{1}
\end{matrix}}
&
\lr{\begin{matrix}
\upw{t_5}{3} & \upw{t_6}{3} \\
\upw{t_5}{2} & \upw{t_6}{2}
\end{matrix}}
\\\cline{2-3}
\begin{matrix}
\newWidth{0} & \newWidth{0} \\
\newWidth{0} & \newWidth{0}
\end{matrix}
&
\begin{matrix}
\upw{t_3}{0} & \upw{t_4}{0} \\
0 & 0
\end{matrix}
&
\lr{\begin{matrix}
\upw{t_5}{1} & \upw{t_6}{1} \\
\upw{t_5}{0} & \upw{t_6}{0} 
\end{matrix}}
\end{vmatrix}.
\label{explicit}
\ee
For compactness, underlined powers signify the factorials in the definition~\eqref{cvf}. Grouping by equal entries gives three $2\times 2$ blocks, shown as diagonal windows. Obviously, only rows $1$--$3$ can be shuffled into the first two positions, because rows $3$--$6$ would introduce a row of zeros in the first block. The sixth row cannot be shuffled into the second block either, for the same reason. Finally, any row can be shuffled into the last block. Because of these restrictions, the Laplace expansion contains only
\be
\binom{3}{2}\binom{5-2}{2}\binom{6-2-2}{2} = 9
\ee
row-blocks:
\be
\begin{split}
[2\,2\,4\,4\,5\,5] &=(1\,2|3\,4|5\,6)-(1\,2|3\,5|4\,6)+(1\,2|4\,5|3\,6)\\
&-(1\,3|2\,4|5\,6)+(1\,3|2\,5|4\,6)-(1\,3|4\,5|2\,6)\\
&+(2\,3|1\,4|5\,6)-(2\,3|1\,5|4\,6)+(2\,3|4\,5|1\,6).
\end{split}
\label{rowdec}
\ee
The main idea of decoding a row-block into Schur functions is to translate row numbers into the powers appearing in the rows. According to Eq.~\eqref{explicit}, the row numbers $(1,2,3)$ translate into powers $(2,1,0)$, respectively, when they appear in the first block. Similarly, $(1,2,3,4,5)$ translate into $(4,3,2,1,0)$ in the second block, and finally $(1,2,3,4,5,6)$ into $(5,4,3,2,1,0)$ in the third block. The first number in each descending sequence is just the value of the respective equal entries $n_i$ in the form. Choosing a particular row-block as an example, this decoding gives
\be
(1\,2|4\,5|3\,6)=|2\,1|1\,0|3\,0|,
\label{decode}
\ee
where parentheses have been changed into vertical bars to signify that the number entries denote powers now. The expression means that rows $1$ and $2$ have powers $2$ and $1$ in the first block, rows $4$ and $5$, shuffled into positions $3$ and $4$, have powers $1$ and $0$ in the second block, and finally rows $3$ and $6$ have powers $3$ and $0$ in the third block, when shuffled into the last two positions.

The powers on the right-hand side of Eq.~\eqref{decode} are those appearing in the numerator of the corresponding Schur function. All that is needed to make the latter appear explicitly is to factor out the Vandermonde denominator from the corresponding minor, e.g. for the first one in $|2\,1|1\,0|3\,0|$:
\be
|21|=\frac{s_{[1,1]}(t_1,t_2)}{2!1!}\Delta_2(t_1,t_2),
\ee
where the factorials are as in Eq.~\eqref{cvf}, and $\Delta_k$ is the usual Vandermonde form in $k$ variables. The partition labeling the Schur function is obtained by subtracting $(1,0)$ as a vector from the powers $(2,1)$,\cite{Stanley99} which corresponds to factoring out the Vandermonde term.

\begin{table}
\begin{center}
\begin{tabular}{cc}
$
\begin{array}{cccccc}
1 & \vbar{2} & 3 & \vbar{4} & 5 & 6 \\  \hline
2 & 1 & 0 & & & \\
4 & 3 & 2 & 1 & 0 & \\
5 & 4 & 3 & 2 & 1 & 0
\end{array}$\phantom{as} & 
\phantom{df}$\begin{array}{cccccc}
t_1 & \vbar{t_2} & t_3 & \vbar{t_4} & t_5 & t_6 \\ \hline
2 & 1 & 0 & & & \\ \cline{1-2}
4 & \vbar{3} & 2 & 1 & 0 & \\ \cline{3-4}
5 & 4 & 3 & \vbar{2} & 1 & 0 \\ \cline{5-6}
\end{array}
$
\vspace{1ex}\\
(a)\phantom{as} & \phantom{df}(b)
\end{tabular}
\end{center}
\caption{Decoding table for the examples of Eqs.~\eqref{rowdec}, \eqref{decod}, and~\eqref{fulldec}. (a) Row numbers above the line, cf. Eq.~\eqref{rowdec}. (b) Variables above the line, cf. Eq.~\eqref{fulldec}. The vertical bars denote blocks of equal entries. The sign of the permutation of the row (variable) indices is the intrinsic (total) sign. The first (leading) term in Eq.~\eqref{fulldec} is marked by a ``staircase'' which gives the type of the form according to Definition~\ref{deftype}.}
\label{dectab}
\end{table}

This decoding algorithm is efficiently implemented with the aid of the decoding tabli in Table~\ref{dectab} (a). The decoding table is used to decode all row-blocks in the row-block expansion of any cv-form. If $a_1 < a_2 < ... < a_r$ are the distinct entries of the cv-form, then the decoding table is an array of non-negative integers of the form
\be
\begin{array}{ccccc}
a_1 & a_1-1&  \cdots & 1&  0 \\
a_2 & a_2-1&  \cdots & 1&  0 \\
\vdots & \vdots &  & & \\
a_r&  a_r-1 & \cdots  & 1&  0,
\end{array}
\ee
where the entries in the table are left justified. Each table row below the line corresponds to one diagonal block, respectively. Row numbers above the line are mapped to powers, listed in the table rows, depending on the block in which they appear, e.g.\ `$4\,5$' is decoded by the second row (to `$1\,0$') because it is in the second block. One gets
\be
(1\,2|4\,5|3\,6)=|2\,1|1\,0|3\,0|=\frac{s_{[1,1]}(1)}{2!1!}\cdot
\frac{s_{[0]}(2)}{1!0!}\cdot
\frac{s_{[2]}(3)}{3!0!}\cdot
\Delta_2(1)\Delta_2(2)\Delta_2(3),
\label{decod}
\ee
where $s_{[0]}\equiv 1$. The parentheses in the last expression indicate the variable set in which the functions are expanded,
\be
(1)\leftrightarrow(t_1, t_2),\quad (2)\leftrightarrow(t_3, t_4),\quad 
(3)\leftrightarrow(t_5, t_6).
\label{varblocks}
\ee
The three Vandermonde factors could have been surmised directly from Eq.~\eqref{explicit}, because it has equal columns when any two variables in these groups (blocks) are equal. Clearly, the Vandermonde factors are the same for all row-blocks in the expansion.

If the entries are not nondecreasing, the columns can be permuted in nondecreasing order, which reduces to the previous case at the price of an overall sign and a permutation of variable indices. This observation suffices to evaluate the cv-form~\eqref{cvf} in general, by a simple improvement of the decoding table in two steps. First, if the entries are nondecreasing, replacing the row-number $i$ above the line with the variable $t_i$ obviously does not change anything: any permutation of the rows can be read from the indices on the variables. Second, if the entries are in random order, reordering the columns in nondecreasing order of entries carries a permutation of the column variables with it, which is noted above the line, e.g.,
\be
\begin{tabular}{ccc}
$[4\;\;5\;\;3\;\;5\;\;3\;\;2]$ & $\to$ &
$[2\;\;3\;\;3\;\;4\;\;5\;\;5]$
\\
$t_1\;t_2\;t_3\;t_4\;t_5\;t_6$ & $\to$ & 
\phantom{x}$\underline{t_6|t_3\;t_5|t_1|t_2\;t_4}$.
\end{tabular}
\ee
Here the first row ends up being labeled by the column-variable $t_6$, etc. The sign of the row-permutation is now effectively multiplied with that of the column-reordering permutation. Grouping variables into blocks is also explicit. Notably, variable indices within a block of equal entries are kept in their original order. (This rule includes the one before, that row numbers must be increasing in each block of the Laplace expansion.) In brief, the cv-form $[4\,5\,3\,5\,3\,2]$ can be evaluated by the same algorithm as $[2\,3\,3\,4\,5\,5]$ up to the sign of the permutation $(6\,3\,5\,1\,2\,4)$ and the corresponding renaming of variables.

The improved decoding table can be used to generate decoded blocks directly, skipping the step~\eqref{rowdec}. In Table~\ref{dectab} (b), one picks a pair of numbers among the three columns in the first row below the line, then a pair from the remaining columns in the next row, which fixes the last pair. The number of terms picked from each row is determined by the blocking, denoted by vertical bars above the line. The order in which they are picked gives a permutation of the variable indices above the line, whose sign is the sign of the contribution. The variable blocks stay the same for all contributions and yield a common Vandermonde factor. Explicitly,
\be
\begin{split}
[2\,2\,4\,4\,5\,5] &=|2\,1|2\,1|1\,0|
-|2\,1|2\,0|2\,0|+|2\,1|1\,0|3\,0|\\&-|2\,0|3\,1|1\,0|+|1\,0|4\,1|1\,0|
+|2\,0|3\,0|2\,0|\\&-|2\,0|1\,0|4\,0|-|1\,0|4\,0|2\,0|+|1\,0|1\,0|5\,0|.
\end{split}
\label{fulldec}
\ee
With this decoding, the Laplace expansion of the cv-form is the same as evaluation. The notation is compact in the sense that a particular partitioning of variables is understood. To avoid ambiguity, the convention is adopted that an individual row-block is evaluated \emph{unsigned}, i.e.\ both the intrinsic sign, given by the permutation of rows in the first version of the decoding table, and the overall sign, coming from the initial ordering permutation of entries (columns), are ignored in Eq.~\eqref{decod}. This convention means that the signs in Eq.~\eqref{fulldec} are multiples of these two (total signs), as determined directly by the second version of the decoding table.

\subsection{Linearly independent confluent Vandermonde forms}

The operator
\be
\nabla^{(k)}=\sum_{i=1}^N\frac{\partial^k}{\partial t_i^k}
\ee
is a symmetrized derivative. It is well known that it annihilates every cv-form:\cite{Bergeron09}
\be
\nabla^{(k)}[n_1\cdots n_N]=0,\quad 1\le k \le N-1,
\label{vanish}
\ee
which means that the cv-forms are \emph{$S_N$-harmonic polynomials}.

\begin{proof}[Proof of Eq.~\eqref{vanish}.]
By the Leibniz formula, the following pattern for $N=3$ is an identity for any $N$, where $a_{ij}$ and $a_{ij}'$, $i,j=1,\ldots,N$, are $2N^2$ arbitrary algebraic indeterminates:
\be
\begin{vmatrix}
a_{11}' & a_{12} & a_{13} \\
a_{21}' & a_{22} & a_{23} \\
a_{31}' & a_{32} & a_{33}
\end{vmatrix}
+
\begin{vmatrix}
a_{11} & a_{12}' & a_{13} \\
a_{21} & a_{22}' & a_{23} \\
a_{31} & a_{32}' & a_{33}
\end{vmatrix}
+
\begin{vmatrix}
a_{11} & a_{12} & a_{13}' \\
a_{21} & a_{22} & a_{23}' \\
a_{31} & a_{32} & a_{33}'
\end{vmatrix}
=
\begin{vmatrix}
a_{11}' & a_{12}' & a_{13}' \\
a_{21} & a_{22} & a_{23} \\
a_{31} & a_{32} & a_{33}
\end{vmatrix}
+
\begin{vmatrix}
a_{11} & a_{12} & a_{13} \\
a_{21}' & a_{22}' & a_{23}' \\
a_{31} & a_{32} & a_{33}
\end{vmatrix}
+
\begin{vmatrix}
a_{11} & a_{12} & a_{13} \\
a_{21} & a_{22} & a_{23} \\
a_{31}' & a_{32}' & a_{33}'
\end{vmatrix}.
\nonumber
\ee
The left-hand side of this pattern becomes the left-hand side of~\eqref{vanish} by substituting $a_{ij}=t_j^{n_j-i+1}/(n_j-i+1)!$ and $a_{ij}'=a_{ij}^{(k)}=t_j^{n_j-i+1-k}/(n_j-i+1-k)!$. Then all terms on the right-hand side have either two equal rows or a row of zeros, so they are all zero.
\end{proof}

\begin{table}
\begin{center}
\begin{tabular}{cc}
$\begin{array}{cccc}
\vbar{t_1} & t_2 & t_3 & t_4 \\ \hline
1 & 0 & &  \\ \cline{1-1}
\vbar{3} & 2 & 1 & 0 \\ \cline{2-4}
\end{array}$\phantom{as} & \phantom{df}
$\begin{array}{cccc}
t_1 & \vbar{t_2} & t_3 & t_4 \\ \hline
2 & 1 & 0 &  \\ \cline{1-2}
3 & \vbar{2} & 1 & 0 \\ \cline{3-4}
\end{array}$
\end{tabular}
\end{center}
\caption{Decoding tables for the cv-forms in Eq.~\eqref{formex}. The staircases determine the type.}
\label{dectabs}
\end{table}

In the language of invariants,\cite{Sturmfels89} harmonicity~\eqref{vanish} generates syzygies:
\be
\nabla^{(1)} [3\,3\,3\,3] = [2\,3\,3\,3] + [3\,2\,3\,3] + [3\,3\,2\,3] + [3\,3\,3\,2] = 0.
\ee
Because there are no other cv-forms of degree $5$ in $4$ variables, their space seems to be three-dimensional. $S_N$-harmonic polynomials span a vector space of dimension $N!$, whose subspaces of cv-forms of a given degree have dimensions given by coefficients of the $q$-factorial\cite{Humphreys90}
\be
[N]_q!=\prod_{i=1}^{N-1}(1+\cdots +q^{i})=\sum_{d=0}^{N(N-1)/2}T(N,d)q^d,
\label{qfac}
\ee
the so-called Mahonian numbers. One can verify the conjecture above, $T(4,5)=3$:
\be
[4]_q!=q^6+3q^5+5q^4+6q^3+5q^2+3q+1.
\label{mahon4}
\ee
Continuing, $\nabla^{(2)}$ and $(\nabla^{(1)})^2$ generate two syzygies:
\begin{gather}
[1\,3\,3\,3] + [3\,1\,3\,3] + [3\,3\,1\,3] + [3\,3\,3\,1] = 0,
\label{dismiss}\\
[2\,2\,3\,3] + [2\,3\,2\,3] + [2\,3\,3\,2] + [3\,2\,2\,3] + [3\,2\,3\,2] + [3\,3\,2\,2] = 0.
\label{keep}
\end{gather}
The term $5q^4$ in~\eqref{mahon4} raises the hope that the second of these accounts for the basis by itself, but how to know that the first is redundant? With the help of Table~\ref{dectabs}, one finds
\be
\begin{split}
 [1\,3\,3\,3] &= |1|2\,1\,0| - |0|3\,1\,0|,\\
[2\,2\,3\,3] &= |2\,1|1\,0| - |2\,0|2\,0| + |1\,0|3\,0|.
\end{split}
\label{formex}
\ee
Say that $[1\,3\,3\,3]$ is of \emph{type} $(1\,2\,1\,0)$, while $[2\,2\,3\,3]$ is of type $(2\,1\,1\,0)$ (Definition~\ref{deftype} below). The type can be read off as the unique ``staircase'' choice of numbers below the line, underlined in Table~\ref{dectabs}. The idea is to treat it as an aggregate label, which in this case distinguishes the two cv-forms.

Among all cv-forms with nondecreasing entries, those whose entries increase at most by one at any step will always generate a type with nonincreasing entries, in other words lexicographically the largest of all permutations of that type. This property is a consequence of the decreasing sequences in the decoding table. E.g., in Table~\ref{dectab} (b), the type $(2\,1\,2\,1\,1\,0)$ is lexicographically out of order because the cv-form $[2\,2\,4\,4\,5\,5]$ has a size-two step from $2$ to $4$. One can ``smooth it out'' to get the cv-form $[2\,3\,3\,4\,5\,5]$, which indeed has the type $(2\,2\,1\,1\,1\,0)$. Such smoothened cv-forms generate all types which are permutations of the largest one. Hence, one conjectures that the cv-forms in Eq.~\eqref{dismiss} can be dismissed because they do not yield new types.
\begin{definition} A cv-form~\eqref{cvf} is called a \emph{regular form} if its entries can be permuted into nondecreasing order such that successive entries differ at most by one:
\be
0\le n_{\sigma(i)+1}-n_{\sigma(i)} \le 1,\quad \sigma(i)=1,\ldots,N-1,
\ee
for some permutation $\sigma$.
\label{defsmooth}
\end{definition}
\begin{definition}
The \emph{standard permutation} of a cv-form $[n_1\cdots n_N]$ is a permutation of the numbers $(1, 2, \ldots, N)$, with order of entries the same as in the cv-form and ties broken from left to right.

\emph{Example.} The standard permutation of of $[4\,5\,5\,3\,3\,2]$ is $(4\,5\,6\,2\,3\,1)$.
\label{defstand}
\end{definition}
\begin{definition}
The \emph{type} of a cv-form $[n_1\cdots n_N]$ is a list of numbers $(k_1\cdots k_N)$ such that $k_i=n_i-s_i+1$, where $s_i$ are the entries in the standard permutation of the same cv-form. These numbers are called \emph{type entries}.

\emph{Example.} The type of $[4\,5\,5\,3\,3\,2]$ is
$$
(4-4+1,\,5-5+1,\,5-6+1,\,3-2+1,\,3-3+1,\,2-1+1)=(1\,1\,0\,2\,1\,2).
$$
\label{deftype}
\end{definition}
Type entries of non-zero cv-forms are always non-negative, because the cv-form of lowest degree is $[0\,1\,2\,\cdots\,N-1]=1$, whose type is $(0\cdots 0)$. For the same reason, they sum to the degree of the cv-form. The types of the forms in Eq.~\eqref{keep} are
\be
\begin{tabular}{cccccc}
[2\,2\,3\,3]& [2\,3\,2\,3]& [2\,3\,3\,2]& [3\,2\,2\,3]& [3\,2\,3\,2]& [3\,3\,2\,2]\\
(2\,1\,1\,0)& (2\,1\,1\,0)& (2\,1\,0\,1)& (1\,2\,1\,0)& (1\,2\,0\,1)& (1\,0\,2\,1)
\end{tabular},
\label{formtype}
\ee
respectively, strenghtening the conjecture that basis sets are built from distinct types, because exactly five are found.

The term $6q^3$ in Eq.~\eqref{mahon4} illustrates a final issue, which concludes the description of the task to generate all and only linearly independent cv-forms explicitly. Two distinct lexicographically greatest types of degree $3$ are $(2\,1\,0\,0)$ and $(1\,1\,1\,0)$. Each contributes three independent cv-forms, so the total space is spanned by
\be
[3\,2\,2\,2],\quad [2\,3\,2\,2],\quad [2\,2\,3\,2],\quad
[3\,3\,2\,1],\quad [3\,2\,3\,1],\quad [3\,2\,1\,3].
\label{six}
\ee
Succintly, $6=3+3$, so the types refine the Mahonian numbers.

\subsection{Classification of confluent Vandermonde forms}

\begin{definition}
A \emph{connected} Young diagram is a Young (skew) diagram\cite{Stanley99} in which it is possible to draw a continuous path visiting all box centers crossing only box edges shared by two adjacent boxes.
\end{definition}

\begin{definition}
A \emph{ribbon} is a connected Young (skew) diagram which does not contain any $2\times 2$ arrangement of boxes.
\end{definition}

\begin{definition}
We define the \emph{class} of a regular form to be its type, permuted in nonincreasing order. The numbers appearing in the class are called \emph{class entries}.

\emph{Note.} Class and type coincide for regular forms with nondecreasing entries.

\emph{Example.} The class of $[4\,5\,5\,3\,3\,2]$ is $(2\,2\,1\,1\,1\,0)$.
\label{defclass}
\end{definition}

\begin{lemma}
Classes of regular forms in $N$ variables are in one-to-one correspondence with ribbons of $N$ boxes.
\label{bijection}
\end{lemma}

\begin{proof}
Represent a class by a regular form with nondecreasing entries $n_i$. By Definition~\ref{deftype}, the connection with class entries $k_i$ is given by
\be
n_i=k_i+i-1.
\ee
Evidently $n_1=k_1$. Also $(k_N,n_N)=(0,N-1)$, because the last of the nondecreasing entries in the cv-form must be $N-1$. Place $N$ boxes at coordinates $(k_i, n_i)$, $i=1,\ldots,N$. Consider them as a Young diagram. It is connected because the form is regular, so as $i\rightarrow i+1$, the $n_i$ increase by one as long as the $k_i$ stay constant, and stay constant as long as the $k_i$ decrease by one. It is a ribbon because $k_i$ cannot increase, and $n_i$ cannot decrease. It has $N$ boxes because exactly one of the coordinates $(k_i,n_i)$ changes as $i\rightarrow i+1$. If $k_i=0$ denotes the first row, and $n_i=N-1$ the last column, the diagram is in the English convention, with the upper-right box always having the coordinates $(0, N-1)$. The reasoning is illustrated in Fig.~\ref{figbijection}.

For the inverse mapping, identify the upper-right box of any connected $N$-box ribbon in the English convention with $(k_N,n_N)=(0,N-1)$. Read off the type and form backwards, $i\rightarrow i-1$, by increasing $k_i$, while keeping $n_i$ constant, when the ribbon goes down, and decreasing $n_i$, while keeping $k_i$ constant, when the ribbon goes left. The $n_i$ and $k_i$ belong to a regular form in $N$ variables and its class.
\end{proof}

\begin{figure}
\centerline{
\includegraphics[width=4cm]{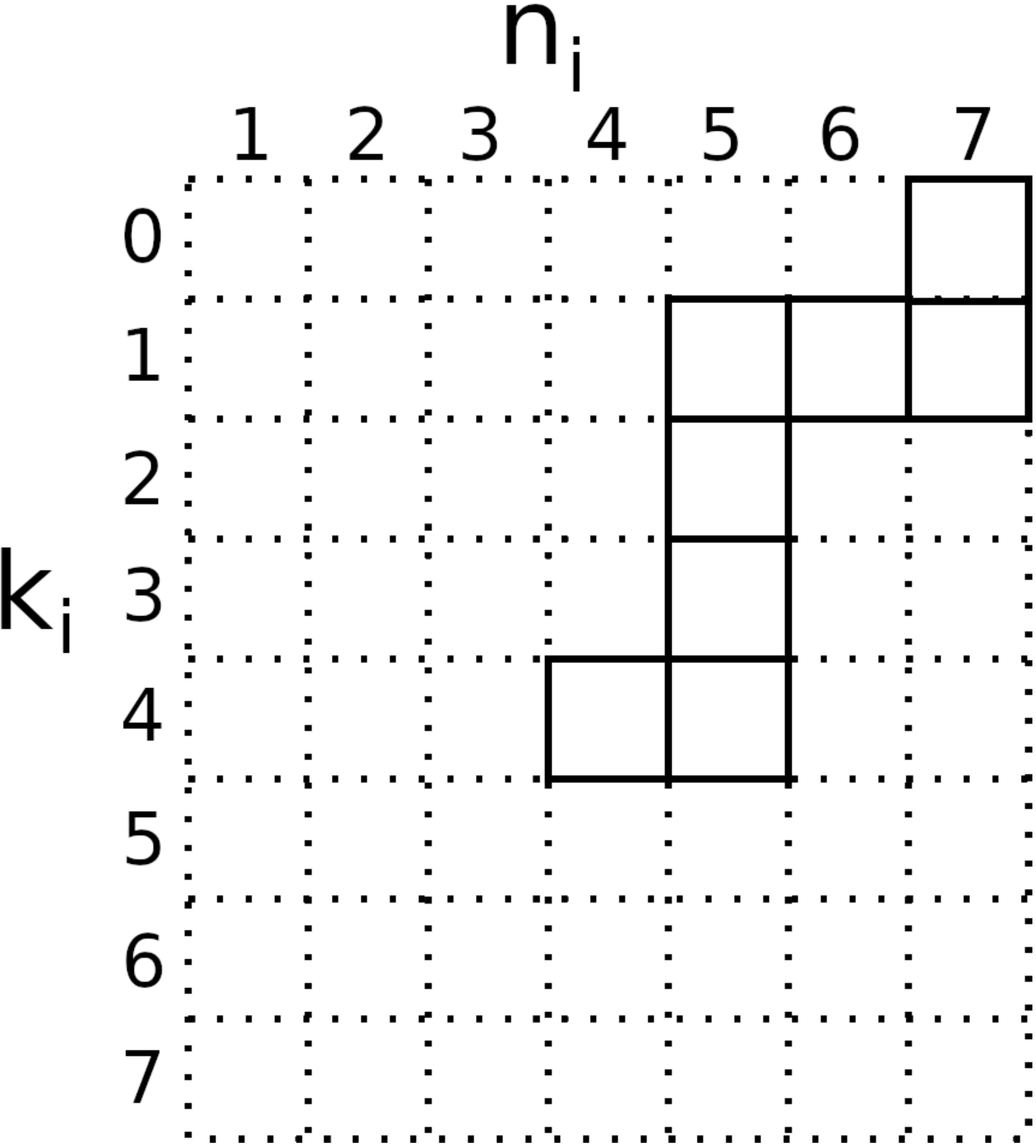}
}
\caption{Illustration of Lemma ~\ref{bijection}, for the class $(4\,4\,3\,2\,1\,1\,1\,0)$ and regular form $[4\,5\,5\,5\,5\,6\,7\,7]$. The class (cv-form) entries list the row (column) coordinates of boxes, respectively. The ribbon represents the partition $4^2321^30$ of $16$, which is the degree of the cv-form. All ribbons with $8$ boxes fit inside the grid except the one encoding the class $(0\,0\,0\,0\,0\,0\,0\,0)$ and cv-form $[0\,1\,2\,3\,4\,5\,6\,7]=1$ (cf. Lemma~\ref{nozero}). All end with a box at position $(0,7)$.}
\label{figbijection}
\end{figure}

The problem how to find all permutations of form entries giving rise to all and only distinct types is solved by the ribbon associated to a class. The solution is to permute the \emph{pairs} $(k_i, n_i)$, simultaneously respecting a constraint on the $n_i$ and avoiding a redundancy in the $k_i$.

The constraint is that variable indices within blocks of equal form entries $n_i$ must be kept in the original order. Because blocks of equal $n_i$ are vertical strips in the diagram, only those permutations are allowed in which the index $i$ increases from top to bottom in the vertical sections.

The redundancy is that permutations of equal entries $k_i$ do not generate new types. It is removed by choosing a representative, namely the one in which the indices $i$ increase from left to right in the horizontal strips of the diagram. Evidently, the two conditions together amount to generating only standard Young skew tableaux from the ribbon diagram. The motivation for the algorithm in the next section is thus established.

\section{Algorithm}

\subsection{Injection from standard ribbon tableaux to cv-forms\label{corr}}

Each standard ribbon tableau encodes a particular cv-form. Label the rightmost column in the ribbon with $N-1$, and columns to the left of it with successively smaller integers. For each tableau, read these horizontal coordinates in the \emph{reversed} order in which the tableau was filled, i.e.\ from the entry $N$ backwards. These are the entries of the corresponding cv-form, e.g. of degree $16$ in $N=8$ variables:
\begin{align}
&
\,\raisebox{4ex}{
\begin{ytableau}
\none & \none & \none & 6 \\
\none & 1 & 2 & 7 \\
\none & 3 \\
\none & 5 \\
4 & 8 \\
\end{ytableau}} \quad\longrightarrow\quad [5\,7\,7\,5\,4\,5\,6\,5].
\label{example}\\
& \quad 4\quad\, 5\quad\, 6\quad 7
\nonumber
\end{align}
Reading backwards resolves a conflict between two conventions, that variable indices \emph{increase} within a block of equal form entries, while type entries \emph{decrease} in any such block. The type can be found similarly, reading the vertical coordinates backwards: $(4\,1\,0\,3\,4\,2\,1\,1)$. [Exercise: which vector in Eq.~\eqref{formtype} is missed by this injection?]

A complementary cv-form of degree $12$ is obtained by reflecting the tableau across the skew-diagonal (\emph{flipping} it) and replacing each entry $k$ with $N-k+1$:
\begin{align}
&
\,\raisebox{4ex}{
\begin{ytableau}
\none & \none & \none & 2 & 3 \\
\none & \none & \none & 7 \\
1 & 4 & 6 & 8 \\
5 \\ 
\end{ytableau}} \quad\longrightarrow\quad [6\,6\,5\,3\,4\,7\,6\,3].
\label{compexample}\\
&\quad 3\quad\, 4\quad 5\quad\, 6\quad 7
\nonumber
\end{align}

This mapping can evidently be applied to arbitrary tableau fillings with distinct entries $1,\ldots,N$ as well, which motivates
\begin{definition}
A \emph{standard form} is a regular form encoded by a standard ribbon tableau, by the injection described above.
\end{definition}

Taking for granted that this encoding maps all standard ribbon tableaux with $N$ boxes onto a complete basis of cv-forms in $N$ variables, the following theorem is a combinatorial interpretation of the symmetry of the $q$-factorial.

\begin{theorem}
There is a bijection between standard ribbon tableaux with $N$ boxes encoding cv-forms of degree $d$ and those encoding cv-forms of degree $N(N-1)/2-d$. It is realized by reflecting a given standard ribbon tableau across the skew-diagonal and replacing each entry $k$ with $N-k+1$.
\end{theorem}

\begin{proof}
The reflection along the skew diagonal changes both increasing directions of tableau entries (down and right) into decreasing ones (left and up), which accounts for the map $k\mapsto N-k+1$. No diagram is mapped onto itself, because whether the first box is part of a horizontal or vertical strip, the reverse is true for the flipped diagram. [In particular, when $N(N-1)/2$ is even, the coefficient of $q^{N(N-1)/4}$ in the $q$-factorial must be even, cf. Eq.~\eqref{six}.] Thus the correspondence is one-to-one, if one can show that the flipped ribbon encodes cv-forms of the complementary degree.

The degree is the sum of class entries:
\be
d=\sum_{k=0}^{l} m_kk,
\label{degree}
\ee
where $m_k$ are the multiplicities of distinct terms (``weights'' $k$, see Fig.~\ref{figbijection}). Here $l+1$ is the height of the diagram (number of its rows). The height of the flipped diagram is equal to the width (number of columns) of the original one:
\be
m_0 + \sum_{k=1}^{l} (m_k - 1) = \sum_{k=0}^{l}m_k - l = N-l,
\label{width}
\ee
taking into account that successive strips of width $m_k$ overlap. The claim is that the respective degrees add to the degree of the Vandermonde form:
\be
\sum_{k=0}^{l}m_kk+\sum_{k=0}^{N-l-1}m_k'k=\frac{N(N-1)}{2},
\label{Nbasis}
\ee
where $m_k'$ are the multiplicities in the flipped diagram. The proof is by induction on $N$. For $N=2$, in obvious shorthand,
\be
\ytableausetup{boxsize=3mm}
d\left(\,
\raisebox{-1pt}{\ydiagram{2}}
\,\right)+
d\left(\,
\raisebox{3pt}{\ydiagram{1,1}}
\,\right)
= 0+1 =\frac{2\cdot 1}{2},
\ytableausetup{boxsize=normal}
\ee
which establishes the basis (the $q$-factorial is $1+q$). It remains to show that the left-hand side of Eq.~\eqref{Nbasis} increases by $N$ when a box is added to the diagram. It can be added in four ways, pairwise equivalent under exchange of the original and flipped diagrams. The two distinct ways are:
\begin{enumerate}
\item Add a box horizontally to the first (zero-weight) row, $m_0\rightarrow m_0+1$. The first sum in Eq.~\eqref{Nbasis} does not change. The new box is alone in the first row of the flipped diagram, which shifts the original flipped diagram downwards by a row, increasing all its weights $k\rightarrow k+1$. The increment of the second sum is the (original) total number of boxes, as required:
\be
\sum_{k=0}^{N-l-1}m_k'\cdot 1=N.
\nonumber
\ee
\item Add a box horizontally to the last (highest-weight) row, $m_l\rightarrow m_l+1$. Then the first sum has an extra term $l$, while the flipped diagram acquires a new highest-weight row with a single box of weight $N-l$, so the total increment is $l+N-l=N$, again as required.
\end{enumerate}
\end{proof}

\subsection{Description of the algorithm\label{algorithm}}

Let the $q$-factorial be as in Eq.~\eqref{qfac}. The following steps generate $T(N,d)$ linearly independent cv-forms of degree $d$ in $N$ variables.
\begin{enumerate}
\item Find all partitions $k_1\ge k_2\cdots \ge k_N=0$ of the number $d$ with at most $N-1$ parts, such that $k_i-k_{i+1}\le 1$ and $k_N=0$. (See Definition~\ref{defclass}.)
\item For each such partition, draw a ribbon Young diagram with $N$ boxes at matrix (row--column) coordinates $(k_i,k_i+i-1)$. (See Fig.~\ref{figbijection}.)
\item For each ribbon diagram, generate all standard Young skew tableax.
\item Map each tableau to a cv-form as described in Section~\ref{corr}.
\end{enumerate}
The numbers $T(N,d)$ give the total number of distinct types of cv-forms of length $N$ and degree $d$. There is usually more than one class of a given degree, which is a refinement of the coefficients $T(N,d)$, as shown in detail in the following example.

\subsection{Example}

Take $N=8$ variables, so that
\be
[8]_q! = q^{28} + \cdots + 3450 q^{16} + \cdots + 3450 q^{12} + \cdots + 1,
\ee
where two complementary terms have been chosen for illustration. The classes belonging to $d=16$ (Step 1) are
\begin{multline}
(54321100),\quad (44322100),\quad (44321110),\quad (43332100),\\
(43322110),\quad (43222210),\quad (33332110),\quad (33322210).
\label{classes}
\end{multline}
Zeros take part in the difference restriction, so that e.g.\ $(33222220)$ is not permitted.

The Young diagram corresponding to $(44321110)$ is shown in Fig.~\ref{figbijection} (Step 2). The number of standard Young tableaux $\#SYT(\lambda,\mu)$ generated by it (Step 3) is given in terms of its skew partition $\lambda/\mu=(44222)/(3111)$ by classical formulas:\cite{Stanley99}
\be
\#SYT(44222,3111) = 315.
\ee
Repeating the exercise for all partitions listed in~\eqref{classes}, one gets respectively
\be
105 + 589 + 315 + 315 + 1385 + 181 + 245 + 315 = 3450,
\label{decomp}
\ee
which is the above-mentioned refinement.

The classes contributing to the complementary term $3450q^{12}$ are
\begin{multline}
 (43221000),\quad (43211100),\quad (33321000),\quad (33221100),\\
 (33211110),\quad (32222100),\quad (32221110),\quad (22222110),
\label{flipclasses}
\end{multline}
where the class $(32222100)$ with $\lambda/\mu=(5441)/(33)$ is complementary to $(44321110)$, see Eqs.~\eqref{example} and~\eqref{compexample}. The list of dimensions is the reverse of Eq.~\eqref{decomp}, because the flipped diagrams generate the same number of tableaux.

\section{Proofs}

\subsection{Counting the basis}

A textbook result is that there are $N!$ standard ribbon tableaux, because each can be indexed by a permutation, whose values are rising as the ribbon goes right, and falling as it goes up. Together with linear independence (next section), this observation is sufficient to show that we have a complete basis. In the remainder of this section, two useful results are given, related to the more fine-grained observation~\eqref{decomp}.

\begin{definition}
The \emph{index} of a ribbon is the degree of the cv-forms it encodes, as given by formula~\eqref{degree}.
\end{definition}

\begin{lemma}
Let $l+1$ be the height of a ribbon of index $d$ with $N$ boxes. There is a bijection between such diagrams and partitions of $d-l(l+1)/2$ into at most $N-l-1$ parts with no part greater than $l$.
\label{classbiject}
\end{lemma}

\begin{proof}
In formula ~\eqref{degree}, the multiplicities $m_k\ge 1$, because the ribbon is connected. Writing $m_k'=m_k-1$, one has
\be
d=\frac{l(l+1)}{2}+\sum_{k=0}^{l} m_k'k,\quad m_k'\ge 0.
\ee
The number of nonzero parts in the remaining sum is $\sum_{k\ge 1} m_k'\le\sum_{k\ge 0} (m_k-1)=N-l-1$, while their sizes $k$ are evidently limited by $l$.
\end{proof}

\emph{Example.} The partitions corresponding to the diagrams in Eqs.~\eqref{example} and~\eqref{compexample} are $6=4+1+1$ and $6=2+2+2$, respectively. They can be read off as the successive row-coordinates of boxes remaining after the first one in each row is erased.

As a corollary, observe that the numbers $l$ and $d$ satisfy the inequalities
\be
d_<(l):=\frac{l(l+1)}{2}\le d \le \frac{l(2N-l-1)}{2}=:d_>(l),
\label{ineq}
\ee
where the upper bound is just $l(l+1)/2+l(N-l-1)$. The two parabolas $d_{<,>}$ intersect at $(l,d)=(0,0)$ and $(N-1,N(N-1)/2)$, the range of the $q$-factorial, so both bounds are tight. An optimized implementation of Step~1 of the algorithm is to find all $(l,d)$ pairs satisfying these inequalities, and apply Lemma~\ref{classbiject}. Several diagrams may correspond to the same $(l,d)$ pair, e.g.\ there are five classes with $l=4$ in Eq.~\eqref{classes}. This refinement is the subject of the next lemma.

\begin{lemma}
The number of $N$-box ribbons of index $d$ and height $l+1$ is given by the coefficient of $q^dt^l$ in the generating function
\be
R(N;q,t)=\sum_{k=0}^{N-1} \qbinom{N-1}{k}{q}q^{\frac{k(k+1)}{2}}t^k
=\prod_{k=1}^{N-1}(1+q^kt).
\ee

\emph{Example.} $R(8;q,t)=q^{28}t^7+\cdots+q^{16}(t^5 + 5t^4 + 2t^3)+ \cdots + q^{12}(2t^4 + 5t^3 + t^2)+ \cdots + 1$, cf.\ Eqs.~\eqref{classes} and ~\eqref{flipclasses}.
\end{lemma}

\begin{proof}
Let $p(L,M;n)$ be the number of partitions of $n$ into at most $M$ parts with no part greater than $L$. These numbers are generated by the $q$-binomial $\qbinom{L+M}{L}{q}$, so $R(N;q,t)$ can be calculated from Lemma~\ref{classbiject}:
\begin{align}
R(N;q,t)&=\sum_{l,d}t^lq^dp(l,N-l-1;d-\frac{l(l+1)}{2})\nonumber\\
&=\sum_{l,d}t^lq^d \mathrm{coeff}[q^{d-\frac{l(l+1)}{2}}] \qbinom{N-1}{l}{q}\nonumber\\
&=\sum_lt^lq^{\frac{l(l+1)}{2}}\sum_d q^{d-\frac{l(l+1)}{2}} \mathrm{coeff}[q^{d-\frac{l(l+1)}{2}}]\qbinom{N-1}{l}{q}\nonumber\\
&=\sum_lt^lq^{\frac{l(l+1)}{2}}\qbinom{N-1}{l}{q}.\nonumber
\end{align}
The equality with the product is the well-known $q$-binomial theorem.
\end{proof}

Putting $q=t=1$ recovers the textbook result that there are $2^{N-1}$ ribbons, because there are $N-1$ coin flips to decide between going right or up.

\subsection{Linear independence}

The main idea is to use the types analogously to leading monomials in a sequence of polynomials graded by degree. The complication is that they do not refer to a single monomial, but to a polynomial term in the row-block expansion of the cv-form.
\begin{definition}
The \emph{row-block order} sets the greater of two row-blocks in the expansion of a given cv-form to be the one with the smaller number of entries $k$, if they have an equal number of entries from $0$ to $k-1$. If their entries are a permutation of each other, the greater is the one later in the lexicographic order.

\emph{Example.} The ordering of terms in Eq.~\eqref{fulldec}.
\end{definition}
The row-block order for a given cv-form is evidently total. It is a well-ordering because the entries in the rows of the decoding table decrease to zero, and terms smaller in the row-block order correspond to moving rightwards in these finite rows. Intuitively, the larger row-blocks distribute available powers more evenly among the variables, so they achieve a greater variability of terms for a given total degree.
\begin{definition}
The \emph{leading row-block} of the expansion of a given cv-form in row-blocks is the largest row-block in the row-block order.

\emph{Note.} For a regular form, the leading row-block maps onto the class when all its vertical bars are erased. Conversely, vertical bars are inserted between equal entries in a class to obtain the leading row-block.

\emph{Example.} The cv-form $[6\,3\,7\,5\,4\,3\,6\,6]$ is of class $(3\,2\,2\,2\,2\,1\,0\,0)$. Its leading row-block is $|3\,2|2|2|2\,1\,0|0|$.
\label{lrblock}
\end{definition}
Identical row-blocks appear in all expansions of cv-forms of a given class, which differ by type. The type-difference is manifested by variable blocks referring to different partitionings of variables. It is easy to extend the row-block order to compare row-blocks belonging to different cv-forms, but it is not necessary here.

\begin{lemma}
The entries in a row-block appear as powers of distinct variables in each monomial in that row-block.

\emph{Example.} $[3221]=-|1|1\,0|0|+|0|2\,0|0|=-t_4(t_2-t_3)+(t_2^2-t_3^2)/2$.
\label{entries}
\end{lemma}
\begin{proof}
The row-blocks are just a notation for the original minors of the Laplace expansion into blocks (the numerators in the classical definition of Schur functions). Their entries denote the powers in each row of a minor, with each entry referring to a distinct variable (column of a minor). When a minor is expanded into monomials, a factor from each row and column appears exactly once in each monomial.
\end{proof}

\begin{lemma}
A collection of standard forms is linearly independent if their leading row-blocks are linearly independent.
\label{lemlin}
\end{lemma}
\begin{proof}
Successive entries in the leading row-blocks of regular forms differ at most by one, by Definition~\ref{defsmooth} (they are equal to the class entries). By Lemma~\ref{entries}, this property (``smoothness'') is inherited in standard forms by the powers in their monomials, when their variables are ordered by increasing order of entries in the cv-form (``entry-ordering'').

Entries in all row-blocks must sum to the same total degree of the cv-form. Increasing one, relative to a leading row-block, requires decreasing another. If smoothness is not lost by this operation, e.g.\ $|3|3\,2|2|2|2\,1\,0|0|$ becomes $|4\,3\,2|2|2\,1|1\,0|0|$, the output is a leading row-block of a different class. Thus every monomial in any lesser row-block of a standard form, with powers ordered by entry-ordering, will contain at least one pair of neighboring variables such that their powers differ by at least two. Choose one with a neighboring pair $t_i^{n+\delta}t_j^{n-\delta}$, say, where $n\ge\delta\ge 1$.

A monomial in a leading row block, in which the same pair appears, must contain at least one variable distinct from $t_i$ and $t_j$ with an exponent between the other two, because it is smooth: $t_k^m$ with $n+\delta>m>n-\delta$ and $k\neq i,j$. The same power $t_k^m$ is missing from the above-chosen lesser row-block by construction. Any other lesser row-block of a standard form will suffer from a similar defect. Thus no linear combination of lesser row-blocks can cancel a 
single leading one, no matter to which standard forms they all belong.
\end{proof}
The row-blocks play a similar role as powers in sequences of polynomial special functions. The lesser ones ensure harmonicity, while the leading term guarantees linear independence.

\begin{definition}
The \emph{characteristic monomial} of a row-block is the monomial obtained by multiplying the diagonal terms in all minors appearing in the row-block.

\emph{Example.} The characteristic monomial of the leading row-block is the product of diagonal terms of the original cv-form, after its entries are permuted in nondecreasing order.
\end{definition}

\begin{lemma}
The characteristic monomial of the leading row-block in a regular form can be read off from the ribbon Young tableau which encodes it, as follows.

Each entry $k$ in a box is associated to the variable $t_{N-k+1}$, raised to the power given by the row-coordinate $m$ of the box: $k\mapsto t_{N-k+1}^m$. The characteristic monomial is the product of these powers.

\emph{Note.} The same mapping gives rise to disconnected ribbons for lesser row-blocks, which is a way to visualize the proof of Lemma~\ref{lemlin}.
\end{lemma}
\begin{proof}
Order the entries in the cv-form nondecreasingly. Observe that the ribbon Young tableau, read from lower left, encodes the diagonal of the cv-form: vertical strips (decreasing powers) correspond to constant entries in the form, while horizontal strips (equal powers) correspond to rising entries in the form. Entries in the tableau correspond to indices on the variables. The characteristic monomial is a product of the diagonal terms.
\end{proof}

\begin{theorem}
The $N!$ standard ribbon tableaux generated by all ribbons with $N$ boxes encode $N!$ linearly independent cv-forms by the injection of Section~\ref{corr}.
\label{central}
\end{theorem}
\begin{proof}
The characteristic monomials of standard forms are unique by the standard tableau filling rules. The horizontal strips select only one permutation of equal powers, while the vertical strips select only one permutation of variables (columns) belonging to equal entries in the cv-form.

A collection of unique monomials is linearly independent. It follows that the leading row-blocks are linearly independent, because extracting the diagonal term is a linear bijection. It is invertible because any alternant can be reconstructed from its diagonal (reading each $1$ as the respective $t_i^0$). Therefore, the standard forms are also linearly independent, by Lemma~\ref{lemlin}.
\end{proof}
This proof depends on the limitation to standard Young tableaux. Relaxing that limitation leads to straightening formulas,\cite{Sturmfels08} expressing an arbitrary cv-form in terms of those encoded by the standard tableaux, cf.\ Eq.~\eqref{keep}. Further considerations along these lines are beyond the scope of this article.

A corollary of Theorem~\ref{central} is that \emph{any} reading order in the mapping of Section~\ref{corr} will give a basis, not only the backward one, stipulated there. The special property of the backward reading is that it selects distinct types as defined by the diagonal monomial of the cv-form. However, by Lemma~\ref{entries}, any other monomial in the expansion of the leading row-block does equally well as the characteristic one. It only differs by a permutation of variable indices (including zero powers). One would have to redefine the type consistently, in order to obtain distinct types under such an alternative choice, but that is redundant if one only wants a basis. The corollary is explicitly
\begin{theorem}
The algorithm in Section~\ref{algorithm} generates $N!$ bases, each indexed by a permutation of the order of reading the entries in the Young tableau in the injection of Section~\ref{corr}.

\emph{Example.} The backward reading is selected by the permutation $\sigma(k)=N-k+1$.
\label{anyread}
\end{theorem}
Theorem~\ref{anyread} simply confirms that no special choices have spoiled the obvious invariance to renaming the variables. An interesting question is how mutually distinct these bases are, especially for large $N$. For example, there are only six choices of five vectors in Eq.~\eqref{keep}, all linearly independent, while there are $24$ bases for $N=4$ by Theorem~\ref{anyread}.

\section{Discussion}

The present work describes an algorithm to generate all linearly independent derivatives of a Vandermonde form, called cv-forms, in a particular notation. The underlying physical motivation is to describe $N$-fermion wave functions in real (laboratory) space of $d$ dimensions, specifically aiming at the case $d=3$. The purpose of this discussion is to place the present work in this wider context.

The cv-forms studied here, defined on a single set of $N$ distinct variables, correspond to $d=1$. Their degree is lower than the degree of the Vandermonde form, which is the lowest-degree polynomial antisymmetric in all $N$ variables, a necessary condition to describe $N$ identical fermions in $d=1$. Thus they do not have direct physical meaning for the problem of $N$ identical fermions. Instead, they appear as building blocks for the wave functions in $d=3$, as described in the Introduction. In addition to $d=3$, the case $d=2$ has attracted a great deal of attention in the context of the fractional quantum Hall effect, where two-dimensional fermionic wave functions satisfying additional symmetry requirements have to be constructed.\cite{Rozman20,Laughlin90}

The cases with $d>1$ require $d$ sets of $N$ variables, one for each dimension of laboratory space. They can be approached\cite{Sunko16-1} as an algorithmic problem: construct the generating set of all \emph{diagonally alternating} functions in $d$ sets of $N$ variables:
\be
\Psi(t_1,\ldots,v_N)=\mathrm{sign}(\sigma)\Psi(t_{\sigma(1)},\ldots,v_{\sigma(N)}).
\label{alter}
\ee
where the permutation $\sigma$ acts diagonally on all $d$-plets, e.g.\ for $d=3$,
\be
\sigma : (t_i,u_i,v_i) \mapsto (t_{\sigma(i)},u_{\sigma(i)},v_{\sigma(i)}), \quad i=1,\ldots,N.
\ee
For $d=1$, the generating set consists of the Vandermonde form alone.\cite{Stanley99} For $d=3$, as mentioned in the Introduction, the generating set seems to consist\cite{Sunko20} of all symmetrized derivatives $\prod_{a,b,c}\nabla^{(a,b,c)}\mathcal{D}_N$ of the form $\mathcal{D}_N$ in Eq.~\eqref{triple}. This conjecture is contingent upon the choice of coefficient ring restricted to bosonic excitations, as given by Eq.~\eqref{choice}.

The case $d=3$ has been studied previously\cite{Bergeron13} with a different choice of coefficient ring, motivated by extending the $N!$ and $(N+1)^{N-1}$ hypotheses\cite{Haiman03} from $d=2$ to $d=3$. The crucial difference with respect to the present approach is that, in the latter case, the coefficient ring consists of \emph{all} diagonally symmetric functions, which is larger than the ring of bosonic excitations~\eqref{choice}, considered here. In that case, the generalized harmonic functions consist of all diagonally alternating functions $\Psi$ which are annihilated by the symmetrized derivatives, $\nabla^{(a,b,c)}\Psi=0$ for all triplets $(a,b,c)$ such that $a+b+c\le n$, the degree of $\Psi$. Such a generalization of the harmonic condition~\eqref{vanish} is too restrictive for physics, although e.g.\ the solutions of the analogous free Laplace equation may play an important special role, in particular, as we show below, they include the ground-state multiplet.

For example, for $N=2$ and $d=3$, there are three linear generators,
\be
\Psi_1=t_1-t_2,\quad \Psi_2=u_1-u_2,\quad \Psi_3=v_1-v_2,
\label{three}
\ee
and a third-degree one,
\be
\Psi_4=\Psi_1\Psi_2\Psi_3,
\label{fourth}
\ee
which must be considered as independent because of the restricted choice~\eqref{choice} of the coefficient ring,\cite{Sturmfels08} although e.g.\ $\Psi_1\Psi_2$ is a diagonally symmetric function.

The generators $\Psi_{1,2,3}$ satisfy the generalized harmonic condition, but $\Psi_4$ does not [notice that $\Psi_4=\mathcal{D}_2$ from Eq.~\eqref{triple}]:
\be
\nabla^{(1,1,0)}\Psi_4=2\Psi_3\neq 0.
\label{rim}
\ee
As the example shows, in the present (physical) case some symmetrized derivatives do not give zero when acting on the generators, but produce other generators instead. One can organize the complete set of generators in a lattice connected by the non-zero operations of symmetrized derivatives.\cite{Sunko22} The subset of generalized-harmonic generators are the non-zero end-points of that procedure, like a rim, which can thus be naturally connected to the zero polynomial as bottom of the lattice. In the above example, the rim consists of the three generators~\eqref{three}, each connected to $\Psi_4$ with a line representing an operation like~\eqref{rim}. The rim necessarily contains the ground-state multiplet, meaning the generators with minimal total degree, and also all generators of degree one greater, because first-order derivatives annihilate them all by Eq.~\eqref{vanish}. It is an interesting conjecture that these are precisely all generators satisfying the generalized harmonic condition in the physical case.

Finding a practical algorithm to generate a basis of all derivatives of the expression~\eqref{triple} is currently an active research effort. The main result of the present work is an essential step along that way. Technically, it is notable how the row-block order made the proof of linear independence easy. It identified the unique monomials by an essentially qualitative measure of variability. Extending these considerations to three dimensions opens the way to study the geometry of many-fermion Hilbert space by algebraic methods. Hopefully this approach will also attract the attention of the mathematical community.

\section*{Acknowledgments}
Conversations with M.~Primc and D.~Svrtan, and communications with F.~Bergeron, P.~Breiding, and M.~Michalek, are gratefully acknowledged. Thanks are also due to the anonymous referee for several improvements in the presentation.

\section*{Data sharing}
Data sharing is not applicable to this article as no new data were created or analyzed in this study.

\bibliography{confluent_forms.bib}

\end{document}